\documentclass[11pt]{article}
\usepackage{amsmath,amsthm}
\usepackage{amssymb,latexsym}
\usepackage{epsfig}
\usepackage{hyperref}
\usepackage{float}
\usepackage{mathpazo}
\usepackage{fullpage}
\usepackage{enumerate}

\usepackage{color,graphicx}

\mathchardef\mhyphen="2D %hyphen used in math mode

\newtheorem{theorem}{Theorem}

\newtheorem{corollary}[theorem]{Corollary}
\newtheorem{lemma}[theorem]{Lemma}

\newtheorem{definition}[theorem]{Definition}
\newtheorem{claim}[theorem]{Claim}
\newtheorem{fact}[theorem]{Fact}

\newenvironment{proof-sketch}{\noindent{\bf Sketch of Proof:}\hspace*{1em}}{\qed\bigskip}
\newenvironment{proof-idea}{\noindent{\bf Proof Idea:}\hspace*{1em}}{\qed\bigskip}
\newenvironment{proof-of-lemma}[1]{\noindent{\bf Proof of Lemma #1:}\hspace*{1em}}{\qed\bigskip}
\newenvironment{proof-of-proposition}[1]{\noindent{\bf Proof of Proposition #1:}\hspace*{1em}}{\qed\bigskip}
\newenvironment{proof-attempt}{\noindent{\bf Proof Attempt:}\hspace*{1em}}{\qed\bigskip}

\newenvironment{remark}{\noindent{\bf Remark:}\hspace*{1em}}{\bigskip}

\newcommand{\pr}{\text{Pr}}

\newcommand{\tensor}{\otimes}
\def\union{\cup}

\newcommand{\bra}[1]{\langle #1|}
\newcommand{\ket}[1]{|#1\rangle}

\renewcommand{\H}{\mathcal{H}}

\newcommand{\NP}{\textsf{NP}}
\newcommand{\NPC}{\textsf{NP}\mhyphen\textsf{complete}}

\newcommand{\QMA}{\textsf{QMA}}

\newcommand{\QMAoC}{\textsf{QMA$_1$-complete}}

\newcommand{\RR}{\textrm{R}}
\newcommand{\LL}{{\cal L}}

\newcommand{\sat}{\textsc{sat}}
\newcommand{\ksat}{k\mhyphen\textsc{sat}}
\newcommand{\qsat}{\textsc{qsat}}
\newcommand{\kqsat}{k\mhyphen\textsc{qsat}}

\usepackage{color}
\definecolor{gray}{rgb}{0.5,0.5,0.5}

 \author{Andris Ambainis\footnote{Faculty of Computing, University of Latvia, Riga, Latvia. Supported by University
 of Latvia Research Grant ZB01-100 and FP7 Marie Curie Grant PIRG02-GA-2007-224886}
 \and
 Julia Kempe\footnote{Blavatnik School of Computer Science, Tel Aviv University, Tel Aviv 69978, Israel.
Supported by the European Commission under the Integrated Project Qubit Applications (QAP) funded by the IST
directorate as Contract Number 015848, by an Alon Fellowship of the Israeli Higher Council of Academic Research, by an
Individual Research Grant of the Israeli Science Foundation, by a European Research Council (ERC) Starting Grant and by
a Raymond and Beverly Sackler Career Development Chair.}
 \and
 Or Sattath
 \thanks{School of Computer Science and Engineering,
 The Hebrew University, Jerusalem, Israel and Blavatnik School of Computer Science, Tel Aviv University, Tel Aviv, Israel.
 sattath@cs.huji.ac.il.}
 }
\title{\bf A Quantum Lov\'{a}sz Local Lemma}
\begin{document}
\maketitle
\begin{abstract}
The Lov\'{a}sz Local Lemma (LLL) is a powerful tool in probability theory to show the existence of combinatorial
objects meeting a prescribed collection of ``weakly dependent" criteria. We show that the LLL extends to a much more
general geometric setting, where events are replaced with subspaces and probability is replaced with relative
dimension, which allows to lower bound the dimension of the intersection of vector spaces under certain independence
conditions.

Our result immediately applies to the $\kqsat$ problem: For instance we show that any collection of rank $1$ projectors
with the property that each qubit appears in at most $2^k/(e \cdot k)$ of them, has a joint satisfiable state.

We then apply our results to the recently studied model of random $\kqsat$. Recent works have shown that the
satisfiable region extends up to a density of $1$ in the large $k$ limit, where the density is the ratio of projectors
to qubits. Using a hybrid approach building on work by Laumann et al. \cite{laumann09dimer} we greatly extend the known
satisfiable region for random $\kqsat$ to a density of $\Omega(2^k/k^2)$. Since our tool allows us to show the
existence of joint satisfying states without the need to construct them, we are able to penetrate into regions where
the satisfying states are conjectured to be entangled, avoiding the need to construct them, which has limited previous
approaches to product states.
\end{abstract}

\newpage

\section{Introduction and Results}

In probability theory, if a number of events are all independent of one another, then there is a positive (possibly
small) probability that none of the events will occur. The Lov\'{a}sz Local Lemma (proved in 1975 by Erd\"os and
Lov\'{a}sz) allows one to relax the independence condition slightly: As long as the events are ``mostly" independent of
one another and are not individually too likely, then there is still a positive probability that none of them occurs.
In its simplest form it states

\begin{theorem}[\cite{erdos1975problems}]
\label{th:symmetric_lll} Let $B_1, B_2,\ldots, B_n$ be events with $\pr(B_i)\leq p$ and such that each event is
mutually independent of all but $d$ of the others. If $p \cdot e \cdot (d+1) \leq 1$ then  $\pr(\bigwedge_{i=1}^n
B_i^c)
> 0$.
\end{theorem}
%A more formal version of this statement is given in section \ref{sec:qlll}, theorem \ref{th:general_lll}.

The Lov\'{a}sz Local Lemma (LLL) is an extremely powerful tool in probability theory as it supplies a way of dealing
with rare events and of showing that a certain event holds with positive probability. It has found an enormous range of
applications (see, e.g., \cite{alon2004probabilistic}), for instance to graph colorability \cite{erdos1975problems},
lower bounds on Ramsey numbers \cite{spencer_asymptotic_1977}, geometry \cite{mani-levitska_decomposition_1987}, and
algorithms \cite{moser2009constructive}. For many of these results there is no known proof which does not use the Local
Lemma.

One notable application of the LLL is to determine conditions under which a $k$-CNF formula is {\em satisfiable}. If
each clause of such a formula $\Phi$ involves a disjoint set of variables, then it is obvious that $\Phi$ is
satisfiable. One way to see this is to observe that a {\em random assignment} violates a clause with probability
$p=2^{-k}$ and hence the probability that all $m$ clauses are satisfied by a random assignment is $(1-p)^m>0$. But what
if some of the clauses share variables, i.e., if they are ``weakly dependent"? This question is readily answered by
using the LLL:
\begin{corollary}
\label{cor:lll_for_sat} Let $\Phi$ be a $\ksat$ formula in CNF-form. If every variable appears in at most $2^k/(e \cdot
k)$ clauses then $\Phi$ is satisfiable.
\end{corollary}
This corollary follows from Thm.~\ref{th:symmetric_lll} by letting $B_i$ be the event that the $i$-th clause is not
satisfied for a random assignment, which happens with probability $p=2^{-k}$, and noting that each clause depends only
on the $d \leq (2^k/e)-k$ other clauses that share a variable with it. In particular this corollary gives a better
understanding of $\sat$, the prototype $\NPC$ problem in classical complexity theory.

In the last decade enormous advances have been made in the area of {\em quantum complexity}, the theory of easy and
hard problems for a quantum computer. In particular, a natural quantum analog of $\ksat$, called $\kqsat$, was
introduced by Bravyi \cite{bravyi2006efficient}: Instead of clauses we have projectors $\Pi_{1},\ldots,\Pi_{m}$, each
acting non-trivially on $k$ qubits, and we have to decide if all of them can be satisfied jointly. More precisely, we
ask if there is a state $\ket{\Psi}$ on all qubits such that $\Pi_i \ket{\Psi}=0$ for all $1 \leq i \leq m$ (in physics
language: we ask if the system is {\em frustration-free}). This problem\footnote{when defined with an appropriate
promise gap for no-instances} was shown to be $\QMAoC$ for $k\geq 4$ \cite{bravyi2006efficient} and as such has
received considerable attention
\cite{liu_consistency_2006,liu_complexity_2007,beigi_complexity_2007,bravyi_complexity_2008,laumann2009phase,laumann09dimer,bravyi2009bounds}.

Note that the question is easy for a set of {\em ``disjoint"} projectors: If no two projectors share any qubits, then
clearly $\ket{\Psi}=\ket{\Psi_1}\otimes \cdots \otimes \ket{\Psi_m}$ is a satisfying state, where $\ket{\Psi_i}$ is
such that $\Pi_i \ket{\Psi_i}=0$, just like in the case of disjoint $\ksat$. It is thus very natural to ask if there
still is a joint satisfying state when the projectors are ''weakly" dependent, i.e., share qubits only with a few other
projectors. One might speculate that a {\em quantum} local lemma should provide the answer.

Motivated by this question we ask: Is there a quantum local lemma? What will take the role of notions like {\em
probability space, probability, events, conditional probability} and {\em mutual independence}? What properties should
they have? And can we prove an analogous statement to Cor.~\ref{cor:lll_for_sat} for $\kqsat$?

\paragraph{Our results:} We answer all these questions in the positive by first showing how to generalize the notions
of probability and independence in a meaningful way applicable to the quantum setting and then by proving a quantum
local lemma. We then show that it implies a statement analogous to Cor.~\ref{cor:lll_for_sat} for $\kqsat$ with exactly
the same parameters as in the classical case. As we describe later in this section, we then combine our results with
recent advances in the study of random $\qsat$ to substantially widen the satisfiable range and to provide greatly
improved lower bounds on the conjectured threshold between the satisfiable and the unsatisfiable region.

Let us first focus on the conceptual step of finding the right notions of {\em probability} and {\em independence}. In
the quantum setting we deal with vector spaces and the probability of a certain event to happen is determined by its
dimension. It is thus very natural to have the following correspondence of classical and ``quantum" notions, using the
apparent similarity between events and linear spaces:

\begin{definition}\label{def:R} We define the following, in correspondence with the classical notions:
\begin{align*}
\begin{array}{lcl}
\textrm{Probability space } \Omega & \rightarrow & \textrm{Vector space } V\\
\textrm{Event } A \in \Omega & \rightarrow & \textrm{Subspace } A \subseteq V \\
\textrm{Complement } A^c=\Omega \setminus A & \rightarrow & \textrm{Orthogonal subspace } A^\perp\\
\textrm{Probability} \Pr(A) & \rightarrow & \textrm{Relative dimension } \RR(A):=\frac{{\dim A}}{{\dim V}}\\
\textrm{Union and Disjunction } A \vee B, A \wedge B & \rightarrow & A +B=\{a+b|a \in A, b \in B\}, A \cap B\\
\textrm{Conditioning } \Pr(A|B)=\frac{\Pr(A \wedge B)}{\Pr(B)} & \rightarrow &  \RR(A|B):=\frac{\RR(A \cap B)}{\RR(B)}=\frac{\dim(A \cap B)}{\dim(B)}\\
A, B \textrm{ independent } \Pr(A \wedge B)=\Pr(A)\cdot \Pr(B)& \rightarrow & A, B \textrm{ R-independent } \RR(A \cap
B)=\RR(A)\cdot \RR(B)
\end{array}
\end{align*}
\end{definition}
This definition by analogy brings us surprisingly far. It can be verified (see Sec.~\ref{sec:prop}) that many useful
properties hold for $\RR$, like (i) $0 \leq \RR \leq 1$, (ii) monotonicity: $A \subseteq B \Rightarrow \RR(A) \leq
\RR(B)$, (iii) the chain rule (iv) an ``inclusion/exclusion" formula and (v) $\RR(A)+\RR(A^\perp)=1$.

There are, however, two important differences between probability and relative dimension. One concerns the complement
of events. For probabilities, the conditional version of property (v) holds: $\Pr(A|B)+\Pr(A^c|B)=1$. For $\RR$ we can
easily find counterexamples to the statement $\RR(A|B)+\RR(A^\perp|B)=1$ (for instance two non-equal non-orthogonal
lines $A$ and $B$ in a two-dimensional space, where $\RR(A|B)+\RR(A^\perp|B)=0$). It is this property that is used in
most proofs of the local lemma, and one of the difficulties in our proof of a quantum LLL (QLLL) is to circumvent its
use.

The second difference concerns our notion of {\em R-independence}. In probability theory, if $A$ and $B$ are
independent, then so are $A^c$ and $B$. Again, this is not true any more for $\RR$ and easy counterexamples can be
found (see Sec.~\ref{sec:prop}). It is thus important to find the right formulation of a quantum local lemma concerning
mutual independence of events. Keeping these caveats in mind and using our notion of relative dimension, we prove a
general quantum LLL (see Sec. \ref{sec:qlll}), which in its simplest form gives:

\begin{theorem}\label{thm:sym-qlll}
Let $X_1,X_2,\ldots, X_n$ be subspaces, where $\RR(X_i)\geq 1-p$ and such that each subspace is mutually R-independent
of all but $d$ of the others. If $p \cdot e \cdot (d+1) \leq 1$ then  $\RR(\bigcap_{i=1}^n X_i)   > 0$.
\end{theorem}
Note that in contrast to the classical LLL in Thm.~\ref{th:symmetric_lll} which is stated in terms of the ``bad" events
$B_i$, here we are working with the ``good" events. While in the classical case these two formulations are equivalent,
this is no longer the case for our notion of $R$-independence.

An immediate application of our QLLL is to $\kqsat$, where we are able to show the exact analogue of
Cor.~\ref{cor:lll_for_sat}.

\begin{corollary}
\label{cor:qsat_satisfiable} Let $\{\Pi_1,\ldots,\Pi_m\}$ be a $\kqsat$ instance where all projectors have rank $1$. If
every qubit appears in at most $2^k/(e \cdot k)$ projectors, then the instance is satisfiable.
\end{corollary}
It follows by defining (with a slight abuse of notation) subspaces $X_i=\Pi_i^\perp$ of satisfying states for $\Pi_i$.
Noticing that $\RR(X_i)=1-2^{-k}$ and that projectors are mutually R-independent whenever they do not share qubits, and
observing that an equivalent formulation of the $\kqsat$-problem is to decide whether $\dim(\bigcap_{i=1}^m
\Pi_i^{\perp})> 0$, Thm.~\ref{thm:sym-qlll} gives the desired result (see Secs.~\ref{sec:prop} and~\ref{sec:qlll} for
details and more applications to $\kqsat$).

\paragraph{Random $\qsat$:}
Over the past few decades a considerable amount of effort was dedicated to understanding the behavior
of \emph{random $\ksat$ formulas}~\cite{kirkpatrick_critical_1994,mezard2002analytic,mezard_clustering_2005}.
Research in this area has witnessed a fruitful collaboration among computer scientists, physicists and mathematicians,
and is motivated in part by an attempt to better understand the class $\NP$,
as well as some recent surprising applications to hardness of approximation (see, e.g., \cite{FeigeRSAT}).

The main focus in this area is an attempt to understand the \emph{phase transition phenomenon} of random $\ksat$,
namely, the sharp transition from being satisfiable with high probability at low clause density to being unsatisfiable
with high probability at high clause density. The existence of this phase transition at a critical density $\alpha_c$
was proven by Friedgut in 1999 \cite{friedgut_sharp_1999};\footnote{Actually, it is still not known whether the
critical density converges for large $n$; see~\cite{friedgut_sharp_1999} for details on this technical (but nontrivial)
issue.} however only in the $k=2$ case its value is known exactly ($\alpha_c=1$
\cite{chvatal_mick_1992,goerdt_threshold_1992,bollobs_scaling_2001}). A long line of works for $k=3$ have narrowed it
down to $3.52 \leq \alpha_c \leq 4.49$ \cite{kaporis_selecting_2003,hajiaghayi_satisfiability_2003,diaz_new_2008} (with
evidence that $\alpha_{c} \approx 4.267$ \cite{mezard2002analytic}), and in the large $k$ limit it has been shown that
$2^k \ln 2 -O(k) \leq \alpha_c \leq 2^k \ln 2$ \cite{achlioptas2004threshold}.

The quantum analogue of this question, namely understanding the behavior of random $\kqsat$ instances,
has recently started attracting attention. As in the classical case, the motivation here comes
from an attempt to understand $\QMA_1$, the quantum analogue of $\NP$ (of which $\kqsat$ is a complete problem),
as well as the possibility of applications to hardness of approximation, but also from the hope to obtain insight
into phase transition effects in other quantum physical systems.

The definition of a random $\kqsat$ instance is similar to the one in the classical case. Fix some $\alpha>0$. Then a
random $\kqsat$ instance on $n$ qubits of density $\alpha$ is obtained by repeating the following $m=\alpha n$ times:
choose a random subset of $k$ qubits and pick a random rank-$1$ projector on them. An equivalent way to describe this
is to say that we choose a random $k$-uniform hypergraph from the ensemble $G_k(n,m)$, in which $m=\alpha n$
$k$-hyperedges are picked uniformly at random from the set of all possible $k$-hyperedges on $n$ vertices (with
repetitions) and then a random rank-$1$ projector is chosen for each hyperedge.

In a first work on the random $\kqsat$ model, Laumann et al. \cite{laumann2009phase} fully characterize the $k=2$ case
and show a threshold at density $\alpha^q_c=1/2$ using a transfer matrix approach introduced by Bravyi
\cite{bravyi2006efficient}. Curiously, the satisfying states in the satisfiable region are {\em product states}. They
also establish the first lower and upper bounds on a possible (conjectured \cite{bravyi2009bounds}) threshold. In a
recent breakthrough Bravyi, Moore and Russell \cite{bravyi2009bounds} have dramatically improved the upper bound to
$0.574 \cdot 2^k$, {\em below} the large $k$ limit of $\ln 2 \cdot 2^k \approx 0.69 \cdot 2^k$ for the classical
threshold!

Recently, Laumann et al. \cite{laumann09dimer} have given substantially improved lower bounds, essentially showing the
following.
\begin{theorem}\cite{laumann09dimer}\label{thm:dimer}
If there is a matching of projectors to qubits such that (i) each projector is matched to a qubit on which it acts
nontrivially and (ii) no qubit is matched to more than one projector, then the $\kqsat$ instance is satisfiable.
\end{theorem}
Such a matching exists with high probability for random instances of $\qsat$ if the density is below some critical
value $c(k)$ (hence $c(k) \leq \alpha^q_c$), with $c(3)\approx 0.92$ and $c(k) \rightarrow 1$ for large $k$.

There remained a distressingly large gap between the best rigorous lower ($<1$) and upper ($\approx 0.574 \cdot 2^k$)
bounds for a satisfiable/non-satisfiable threshold of random $\kqsat$.

Using our quantum LLL we are able to dramatically improve the lower bound on such a threshold. To get a better
intuition on the kind of bounds the quantum LLL can give in this setting, let us first look at a simple toy example:
random $\kqsat$ instances picked according to the uniform distribution on {\em $D$-regular} $k$-hypergraphs $G_k(n,D)$
(so $m=Dn/k$ and their density is $\alpha=D/k$). It is easy to see that a matching as assumed in Thm.~\ref{thm:dimer}
only exists iff $k\geq D$, so this technique shows satisfiability only below density $1$. Our
Cor.~\ref{cor:qsat_satisfiable}, on the other hand, immediately implies that the instance is satisfiable as long the
density $\alpha \leq 2^k/(e \cdot k^2)$. It is this order of magnitude that we manage to achieve also in the random
$\kqsat$ model described above. We show

\begin{theorem}\label{thm:threshold}
A random $\kqsat$ instance of density $\alpha \leq  2^k/(12 \cdot e \cdot k^2)$ is satisfiable with high probability
for any $k\geq 1$. Hence $\alpha_c^q \geq 2^k/(12 \cdot e \cdot k^2)$.
\end{theorem}
All previous lower bound proofs \cite{laumann2009phase,laumann09dimer} were based on constructing {\em tensor product
states} which satisfy all constraints. In fact it is conjectured \cite{laumann09dimer} that $c(k)$ is the critical
density above which entangled states would necessarily appear as satisfying states. To our knowledge no technique has
allowed to deal with entangled satisfying states in this setting. Using the quantum LLL allows us to show the existence
of a satisfying state without the need to generate it, and in particular the satisfying state need not be a product
state (and probably is not). We conjecture that the improvement in our bound, which is roughly exponential in $k$, is
due to this difference.

The main difficulty we encounter in the proof of Thm.~\ref{thm:threshold} (see Sec.~\ref{sec:random}) is that even
though the {\em average degree} in $G_k(n,m=\alpha n)$ is of the right order of magnitude ($\approx 2^k/k$) to apply
the quantum LLL (Cor.~\ref{cor:qsat_satisfiable}), the maximum degree can deviate vastly from it (its expected size is
roughly logarithmic in $n$), and hence prevent a direct application of the quantum LLL. The key insight is that we can
split the graph into two parts, one essentially consisting of {\em high degree} vertices that deviate by too much from
the average degree and the other part containing the remaining vertices. We then show that the first part obeys the
matching conditions of Thm.~\ref{thm:dimer} \cite{laumann09dimer} and hence has a satisfying state, and the second part
obeys the maximum degree requirements of the quantum LLL and is hence also satisfiable. The challenge is to ``glue"
these two satisfying solutions together. For this we need to make sure that each edge in the second part intersects the
first part in at most one qubit (by adding all other edges to the first part, while carefully treating the resulting
dependencies). We can then create a new $(k-1)$-local projector of rank $2$ for each intersecting edge, which reflects
the fact that one qubit of this edge is already ``taken". This allows to effectively decouple the two parts.

\paragraph{\bf Discussion and Open Problems:} We have shown a general quantum LLL. An obvious open question is whether it has more
applications for quantum information.

We call our generalization of the Lov\'{a}sz Local Lemma ``quantum" in view of the applications we have given. However,
stricto sensu there is nothing {\em quantum} in our version of the LLL; It is a statement about subspaces and the
dimensions of their intersections. As such it seems to be very versatile and we hope that it will find a multitude of
other applications, not only in quantum information, but also in geometry or linear algebra. More generally, our LLL
holds for any set of objects with a valuation $\RR$ and operations $\bigcap$ and $+$ that obey properties (i)-(iv) (see
Lemma~\ref{lem:propR}) and might be applicable even more generally. Since the LLL has so many applications, we hope
that our ``geometric" LLL becomes equally useful.

The standard proof of the classical LLL is {\em non-constructive} in the sense that it asserts the existence of an
object that obeys a system of constraints with limited dependence, but does not yield an efficient procedure for {\em
finding} an object with the desired property. In particular, it does not provide an efficient way to find the actual
satisfying assignment in Cor.~\ref{cor:lll_for_sat}. A long line of research
\cite{beck_algorithmic_1991,alon_parallel_1991,molloy_further_1998,czumaj_coloring_2000,srinivasan_improved_2008,moser_derandomizinglovasz_2008}
has culminated in a very recent breakthrough result by Moser \cite{moser2009constructiveSTOC} (see also
\cite{moser2009constructive}), who gave an {\em algorithmic} proof of the LLL that allows to efficiently {\em
construct} the desired satisfying assignment (and more generally the object whose existence is asserted by the LLL
\cite{moser2009constructive}). Moser's algorithm itself is a rather simple random walk on assignments; an innovative
information theoretic argument proves its correctness (see also \cite{fortnow_kolmogorov}). This opens the exciting
possibility to draw an analogy for a (possibly quantum) algorithm to construct the satisfying state in instances of
$\qsat$ which are known to be satisfiable via our QLLL, and we hope to explore this connection in future work.

\paragraph{Structure of the paper:} In Sec.~\ref{sec:prop} we study properties of relative dimension $\RR$ and
of R-independence, allowing us to prove a general QLLL in Sec.~\ref{sec:qlll}. Sec.~\ref{sec:random} extends our
results to the random $\kqsat$ model and presents our improved bound on the size of the satisfiable region.

\section{Properties of Relative Dimension}\label{sec:prop}

Here we summarize and prove some of the properties of the {\em relative dimension $R$} and of {\em R-independence} as
defined in Def.~\ref{def:R}, which will be useful in the proof of the quantum LLL in the next section.

\begin{lemma}\label{lem:propR} For any subspaces $X,Y, Z, X_i \subseteq V$ the following hold
\begin{enumerate}[(i)]
\item $0 \leq \RR(X) \leq 1$.
 \item Monotonicity: $X \subseteq Y \rightarrow \RR(X) \leq \RR(Y)$.
 \item Chain Rule:\\ $\RR(\bigcap_{i=1}^n X_i|Y)=\RR(X_1|Y) \cdot \RR(X_2|X_1 \cap Y )
 \cdot \RR(X_3|X_1 \cap X_2 \cap Y) \cdot \ldots \cdot \RR(X_n|\bigcap_{i=1}^{n-1}X_i \cap Y)$.
 \item Inclusion/Exclusion: $\RR(X)+\RR(Y) = \RR(X+Y) + \RR(X \cap Y)$.
 \item $\RR(X)+\RR(X^\perp)=1$ and $\RR(X|Y)+\RR(X^\perp|Y)\leq 1$.
 \item $\RR(X|Z)+\RR(Y|Z)-\RR(X \cap Y|Z) \leq 1$.
\end{enumerate}
\end{lemma}
\begin{proof}
Properties (i), (ii), (iii) and (v) follow trivially from the definition.

Property (iv) follows from $\dim(X)+\dim(Y)=\dim(X \cap Y) + \dim (X + Y)$, which is an easy to prove statement about
vector spaces (see e.g. \cite{kostrikin_linear_1997}, Thm. 5.3).

Property (vi) follows from (ii) and (iv): Inclusion/exclusion (iv) gives $\RR(X \cap Z)+\RR(Y \cap Z) =\RR(X \cap Z + Y
\cap Z)+\RR(X \cap Y \cap Z)\leq \RR(Z)+\RR(X \cap Y \cap Z)$, where the last inequality follows from the monotonicity
property (ii) using $X \cap Z + Y \cap Z \subseteq Z$. Dividing by $\RR(Z)$ gives the desired result.
\end{proof}

We also need to extend our definition of R-independence (Def.~\ref{def:R}) to the case of several subspaces, in analogy
to the case of events.

\begin{definition}[Mutual independence] An event $A$ (resp. subspace $X$) is mutually independent (resp. mutually R-independent)
of a set of events (resp. subspaces) $\{Y_1, \ldots, Y_\ell\}$ if for all $ S \subseteq [\ell]$,
$\Pr(A|\bigwedge_{i=1}^\ell Y_i)=\Pr(A)$ (resp. $\RR(X|\bigcap_{i=1}^\ell Y_i)=\RR(X)$).
\end{definition}

Note that unlike in the case of probabilities, it is possible that two subspaces $A$ and $B$ are mutually R-independent
but $A^c$ and $B$ are not mutually R-independent. One example for this are the following subspaces of $\mathbb{R}^4$:
$A=\textrm{span}(\{(1,0,0,0),(0,1,0,0)\}$ and $B=\textrm{span}(\{(1,0,0,0),(0,1,1,0)\}$. We have $\RR(A|B)=\RR(A)=1/2$
but $\RR(A^\perp|B)=0$ while $\RR(A^\perp)=1/2$.

Let us now relate the notion of mutual R-independence to the situation in $\kqsat$ instances. We first associate a
subspace with a projector, in the natural way.

\begin{definition}[Projectors and associated subspace]
A $k$-local projector on $n$-qubits is a projector of the form $\pi \otimes I_{n-k}$, where $\pi$ is a projector on $k$
qubits $q_1, \ldots , q_k$ and $I_{n-k}$ is the identity on the remaining qubits. We say that $\Pi$ acts on
$q_1,\ldots, q_k$. For a projector $\Pi$, let its satisfying space be $X_{\Pi^\perp}:=\ker \Pi=\{\ket{\Psi}\,|\, \Pi
\ket{\Psi}=0\}$. When there is no risk of confusion we denote $X_{\Pi^\perp}$ by $\Pi^\perp$ and its complement by
$\Pi$.
\end{definition}

Recall that in statements like Cor.~\ref{cor:qsat_satisfiable} we would like to say that two projectors are mutually
R-independent if they do not share any qubits. This is indeed the case, as the following lemma shows.

\begin{lemma}\label{lem:tensor-indep}
Assume a projector $\Pi$ does not share any qubits with projectors $\Pi_{1},\ldots, \Pi_{\ell}$. Then $X_{\Pi^\perp}$
is mutually R-independent of $\{X_{\Pi_1^\perp}, \ldots , X_{\Pi_\ell^\perp}\}$.
\end{lemma}
\begin{proof}
Let us split the Hilbert space $\mathcal{H}$ of the entire system into $\mathcal{H}=\mathcal{H}_1\otimes
\mathcal{H}_2$, where $\mathcal{H}_1$ is the space  which consists of the qubits $\Pi$ acts on non-trivially (and
$\Pi_1,\ldots , \Pi_{\ell}$ act as identity) and the remaining space $\mathcal{H}_2$. By assumption there are
projectors $\pi$ and $\pi_1,\ldots,\pi_\ell$ such that $\Pi=\pi \otimes I_{n-k}$ and $\Pi_i=I_k \otimes \pi_i$. For
every $S \subseteq [\ell]$,

\begin{equation*}
\RR(\Pi |\bigcap_{i \in S} \Pi_{i}) = \frac{\dim(\Pi \bigcap_{i\in S} \Pi_{i})}{\dim(\bigcap_{i \in S} \Pi_{i}} =
\frac{\dim(\pi \tensor \bigcap_{i\in S} \pi_{i})}{\dim( I \tensor \bigcap_{i \in S} \pi_{i})} =
\frac{\dim(\pi)\dim(\bigcap_{i \in S} \pi_{i} ) }{\dim(\mathcal{H}_{1})\dim(\bigcap_{i \in S} \pi_{i})} = \RR(\Pi) .
\end{equation*}
\end{proof}
\begin{remark}
In exactly the same way one can show that $\Pi$ is mutually R-independent of $\{\Pi_1^\perp,\ldots ,\Pi_\ell^\perp\}$
and that both $\Pi$ and $\Pi^\perp$ are mutually R-independent of $\{\Pi_1,\ldots ,\Pi_\ell\}$. Hence the property of
not sharing qubits (or, for subspaces, having a certain tensor structure), which in particular implies mutual
R-independence, is in some sense a stronger notion of independence than R-independence. To prove our quantum LLL we
only require the weaker notion of R-independence, which potentially makes the quantum LLL more versatile and applicable
in settings where there is no tensor structure.
\end{remark}

\section{The Quantum Local Lemma}
\label{sec:qlll}

We begin by stating the classical general Lov\'{a}sz Local Lemma. To this end we need to be more precise about what we
mean by ``weak" dependence, introducing the notion of the {\em dependency graph} for both events and subspaces (see
e.g. \cite{alon2004probabilistic} for the case of events), where we use {\em relative dimension} R as in
Def.~\ref{def:R}.

\begin{definition}[Dependency graph for events/subspaces]
\label{def:events_dependency_graph} The directed graph $G=([n],E)$ is a dependency graph for \\(i) the events
$A_1,\ldots,A_n$ if for every $i \in [n]$, $A_i$  is mutually independent of $\{A_j| (i,j) \notin E \}$,\\
(ii) the subspaces $X_1, \ldots, X_n$ if for every $i \in [n]$, $X_i$ is mutually R-independent of $\{X_j| (i,j) \notin
E \}$.
\end{definition}
With these notions in place we can state the general Lov\'{a}sz Local Lemma (sometimes also called the asymmetric LLL).
\begin{theorem}[\cite{erdos1975problems}]
\label{th:asymmetric_lll} Let $A_1,A_2,\ldots,A_n$ be events with dependency graph $G=([n],E)$. If there exists $0 \leq
y_1,\ldots , y_n < 1$, such that  $Pr(A_i)\leq y_i \cdot \prod_{(i,j)\in E} (1-y_j)$, then
\[ \pr(\bigwedge_{i=1}^n A_i^{c}) \geq \prod_{i=1}^n (1- y_i).\]
\end{theorem}
In particular, with positive probability no event $A_i$ holds.

We prove a quantum generalization of this lemma with exactly the same parameters. As mentioned before, we have to
modify the formulation of the LLL to account for the unusual way R-independence behaves under complement. We are now
ready to state and prove our main result.

\begin{theorem}[Quantum Lov\'{a}sz Local Lemma]
\label{th:asymmetric_qlll} Let $X_1,X_2,\ldots,X_n$ be subspaces with dependency graph $G=([n],E)$. If there exist $0
\leq y_1,\ldots , y_n < 1$, such that
\begin{equation}
\label{eq:given_non_symmetric} \RR(X_i)\geq 1-y_i \prod_{(i,j) \in E} (1-y_j),
\end{equation}
then $ \RR(\bigcap_{i=1}^n X_i) \geq \prod_{i=1}^n (1-y_i)$.
\end{theorem}
Note that when $\RR$ is replaced by $\Pr$ and $\bigcap$ by $\bigwedge$ we recover the LLL Thm.~\ref{th:asymmetric_lll}.
Our proof uses properties that hold both for $\Pr$ and $\RR$, in particular we also prove Thm.~\ref{th:asymmetric_lll}.
One can say that we generalize the LLL to any notion of probability for which the properties (i)-(iv) of
Lemma~\ref{lem:propR} hold (these are the only properties of $\RR$ we need in the proof).

\begin{proof}[Proof of Theorem \ref{th:asymmetric_qlll}:]
We modify the proof in \cite{alon2004probabilistic} in order to avoid using the property $\Pr(A|B)+\Pr(A^c|B)=1$ which
does not hold for $\RR$. To show Thm.~\ref{th:asymmetric_qlll}, it is sufficient to prove the following Lemma.

\begin{lemma}
 For any $S \subset [n]$, and every $i \in [n]$, $\RR(X_i|\bigcap_{j \in S} X_j) \geq 1 - y_i$.
\end{lemma}
Thm.~\ref{th:asymmetric_qlll} now follows from the chain rule (Lemma~\ref{lem:propR}.iii):
\[ \RR(\bigcap_{i=1}^n X_i) = \RR(X_1) \RR(X_2|X_1) \RR(X_3|X_1\cap X_2)
\ldots \RR(X_n|\bigcap_{j=1}^{n-1}X_j) \geq \prod_{i=1}^n \left(1-y_i\right). \]

We prove the lemma by complete induction on the size of the set $S$. For the base case, if $S$ is empty, we have
\[\RR(X_i)\geq 1-\left[ y_i \prod_{(i,j)\in E} (1-y_j) \right] \geq 1-y_i. \]
Inductive step: To prove the statement for $S$ we assume it is true for all sets of size $< |S|$. Fix $i$ and define $D
= S\cap \{j|(i,j) \in E\}$ and $I=S \backslash D$ (I and D are the independent and dependent part of S with respect to
the i'th element). Let ${\cal X}_I=\bigcap_{j \in I} X_j$ and  ${\cal X}_D=\bigcap_{j \in D} X_j$. Then
 \begin{equation}\label{eq:equivalent_form}
1-\RR(X_i|\bigcap_{j \in S} X_j)=1-\RR(X_i|{\cal X}_I \cap {\cal X}_D)= 1- \frac{\RR(X_i \cap {\cal X}_D | {\cal
X}_I)}{ \RR({\cal X}_D| {\cal X}_I)}=
 \frac{\RR({\cal X}_D | {\cal X}_I)-\RR(X_i \cap {\cal X}_D | {\cal X}_I)}{ \RR({\cal X}_D| {\cal X}_I)}.
 \end{equation}
To show the lemma we need to upper bound this expression by $y_i$. We first upper bound the numerator:
\begin{eqnarray*}
& &  \RR({\cal X}_D | {\cal X}_I)-\RR(X_i \cap {\cal X}_D | {\cal X}_I) \leq 1- \RR(X_i | {\cal X}_I) = 1-\RR(X_i) \leq
y_i \prod_{(i,j) \in E} (1- y_j),
\end{eqnarray*}
where for the first inequality we use Lemma~\ref{lem:propR}.vi, then the fact that $X_i$ and ${\cal X}_I$ are
R-independent, and the assumption on $\RR(X_i)$, Eq.~\eqref{eq:given_non_symmetric} in Thm.~\ref{th:asymmetric_qlll}.

Now, we lower bound the denominator of Eq.~\eqref{eq:equivalent_form}. Suppose $D=\{j_1,\ldots,j_{|D|} \}$, then
\begin{eqnarray*}
& & \RR\left(\bigcap_{j \in D} X_j | {\cal X}_I\right) =
 \RR\left(X_{j_1}|{\cal X}_I\right) \cdot \ldots \cdot \RR\left(X_{j_{|D|}}|X_{j_1} \cap  \ldots \cap X_{j_{|D|-1}} \cap {\cal
 X}_I\right)
\geq \prod_{j \in D}1-y_j \geq \prod_{j:(i,j) \in E} 1-y_j.
\end{eqnarray*}
The equality follows from the chain rule (Lemma~\ref{lem:propR}.iii), the first inequality follows from the inductive
assumption, and the second inequality follows from the fact that $D=\{j|(i,j) \in E\}\cap S \subseteq \{j|(i,j) \in
E\}$, and that $y_j < 1$.
\end{proof}

For many applications we only need a simpler version of the quantum LLL, often called the symmetric version, which we
have already stated in Thm.~\ref{thm:sym-qlll}.

 \begin{proof}[Proof of Theorem~\ref{thm:sym-qlll}:]
Thm.~\ref{thm:sym-qlll} follows from Thm.~\ref{th:asymmetric_qlll} in the same way the symmetric LLL of
Thm.~\ref{th:symmetric_lll} follows from the more general LLL of Thm.~\ref{th:asymmetric_lll}
\cite{alon2004probabilistic}; we include it here for completeness: If $d=0$ then $\RR(\bigcap_{i=1}^n X_i)=\Pi_{i=1}^n
\RR(X_i)>0$ by the chain rule (Lemma~\ref{lem:propR}.iii) and mutual R-independence of all subspaces. For $d \geq 1$,
by the assumption there is a dependency graph $G = ([n],E)$ for the subspaces $X_1,\ldots,X_n$ in which for each $i$;
$|\{j|(i,j) \in E\}| \leq d$. Taking $y_i=1/(d+1)$ ($<1$) and using that for $d \geq 1$, $(1-\frac{1}{d+1})^d
>\frac{1}{e}$ we get
 $$\RR(X_i) \geq 1-p \geq 1-\frac{1}{e(d+1)}\geq 1-\frac{1}{d+1}(1-\frac{1}{d+1})^d\geq 1-y_i(1-y_i)^{|\{j|(i,j) \in E\}|}, $$
which is the necessary condition Eq.~\eqref{eq:given_non_symmetric} in Thm.~\ref{th:asymmetric_qlll}. Hence
\begin{equation}\label{eq:d-bound}
\RR(\bigcap_{i=1}^n X_i) \geq (1-\frac{1}{d+1})^n>0.
\end{equation}
\end{proof}
Note that Eq.~\eqref{eq:d-bound} also allows us to give a lower bound on the dimension of the intersecting subspace,
which might be useful for some applications.

We can now move to the implications of the QLLL for ``sparse" instances of $\qsat$ and prove
Cor.~\ref{cor:qsat_satisfiable}. It is a special case of this slightly more general Corollary.

\begin{corollary}\label{cor:gen-sat-qlll}
Let $\{\Pi_1,\ldots,\Pi_m\}$ be a $\kqsat$ instance where all projectors have rank at most $r$. If every qubit appears
in at most $D=2^k/(e \cdot r \cdot k)$ projectors, then the instance is satisfiable.
\end{corollary}
\begin{proof}
By assumption, each projector shares qubits with at most $k(D-1)$ other projectors. As we have already shown in
Lemma~\ref{lem:tensor-indep}, each $\Pi_i^\perp$ is mutually R-independent from all but $d=k(D-1)$ of the other
$\Pi_j^\perp$. With $p=r\cdot 2^{-k}$ we have $\RR(\Pi_i^\perp) \geq 1-p$. The corollary follows from
Thm.~\ref{thm:sym-qlll} because $p\cdot e \cdot (d+1) \leq r \cdot 2^{-k}\cdot e(k(2^k/(e \cdot r \cdot k)-1)+1) \leq
1$.
\end{proof}

\section{An improved lower bound for random $\qsat$}\label{sec:random}

This section is devoted to the proof of Thm.~\ref{thm:threshold}. As mentioned in the introduction, in {\em random
$\kqsat$} we study a distribution over instances of $\kqsat$ with fixed density, defined as follows.

\begin{definition}[Random $\kqsat$]
Random $\kqsat$ of density $\alpha$ is a distribution over instances $\{\Pi_1,\ldots ,\Pi_m\}$ on $n$ qubits, where
$m=\alpha n$, obtained as follows:
\begin{enumerate}
\item Construct a $k$-uniform hypergraph $G$ with $n$ vertices and $m$ edges (the constraint hypergraph) by choosing
$m$ times, uniformly and with replacement, from the ${n \choose k}$ possible $k$-tuples of vertices.
 \item For each edge $i$ ($1 \leq i \leq m$) pick a $k$-qubit state $\ket{v_i}$ acting on the corresponding qubits uniformly from all such states
 (according to the Haar measure) and set $\Pi_i=\ket{v_i}\bra{v_i} \otimes I_{n-k}$.
\end{enumerate}
\end{definition}

\begin{remark}{\bf ($G_k(n,m)$ vs. $G_k(n,p)$)} The distribution on hypergraphs obtained in the first step is denoted by $G_k(n,m)$ and has been studied
extensively (see, e.g., \cite{bollobas_random_2001,alon2004probabilistic}). A closely related model is the so called
Erd\"os-Renyi $G_k(n,p)$ model, where each of the ${n \choose k}$ $k$-tuples is independently chosen to be an edge with
probability $p$. For $p=m/{n \choose k}$ the expected number of edges in $G_k(n,p)$ is $m$ and these two distributions
are very close to each other. In most cases proving that a certain property holds in one implies that it holds in the
other (see \cite{bollobas_random_2001}). There seems to be no consensus whether to define the random $\ksat$ and
$\kqsat$ models with respect to the distribution $G_k(n,m)$ or $G_k(n,p)$; for instance the upper bounds on the random
$\kqsat$ threshold of \cite{bravyi2009bounds} are shown in the $G_k(n,m)$ model, whereas the lower bounds
\cite{laumann2009phase,laumann09dimer} are given in the $G_k(n,p)$ model. This, however, does not matter, as properties
such as being satisfiable with high probability will always hold for both models.
\end{remark}

As mentioned, for $\alpha=c \cdot 2^k/k^2$, even though a graph from $G_k(n,m)$ has average degree $D_{avg}=k\alpha=c
\cdot 2^k/k$, and hence {\em on average} each qubit appears in $c \cdot 2^k/k$ projectors, we cannot apply the QLLL and
its Cor.~\ref{cor:qsat_satisfiable} directly: The degrees in $G_k(n,m)$ are distributed according to a Poisson
distribution with mean $D_{avg}$ and hence we expect to see some high degree vertices (in fact the expected maximum
degree at constant density is expected to be roughly logarithmic in $n$ \cite{bollobas_random_2001}). The idea behind
the proof of Thm.~\ref{thm:threshold} is to single out the ``high-degree" part $V_\H$ of the graph and to treat it
separately. The key is to show (i) that the matching conditions of Laumann et al.'s Thm.~\ref{thm:dimer} is fulfilled
by $V_\H$ on one hand and  (ii) to demonstrate how to ``glue" the solution on $V_\H$ with the one provided by QLLL on
the remaining graph.

We first show how to glue two solutions, which also clarifies the requirements for $\H$.

\begin{lemma}[Gluing Lemma]\label{lem:glue}
Let ${\cal P}=\{\Pi_1,\ldots ,\Pi_m\}$ be an instance of $\kqsat$ with rank-$1$ projectors. Assume that there is a
subset of the qubits $V_\H$ and a partition of the projectors into two sets $\H$ and $\LL$, where $\H$ (possibly empty)
consists of all projectors that act only on qubits in $V_\H$, such that
 \begin{enumerate} \item The reduced instance given by $\H$ (restricted to qubits in $V_\H$) is satisfiable.
 \item Each qubit $\notin V_\H$ appears in at most $2^k/(4\cdot e \cdot k)$ projectors from $\LL$.
 \item Each projector in $\LL$ has at most one qubit  in $V_\H$.
\end{enumerate}
Then ${\cal P}$ is satisfiable.
\end{lemma}
\begin{proof}
Let $\ket{\Phi_\H}$ be a satisfying state for $\H$ on the qubits $V_\H$ (if ${\H}=\emptyset$ this can be any state). To
extend it to the whole instance, we need to deal with the projectors in $\LL$ acting on a qubit from $V_\H$. Let
${\LL}=\{\Pi_1, \ldots , \Pi_\ell\}$. From $\LL$ we construct a new ``decoupled" instance ${\LL}'=\{Q_1, \ldots ,
Q_\ell\}$ of $\kqsat$ with projectors of rank at most $2$ that have no qubits in $V_\H$. If $\Pi_i \in \LL$ does not
act on any qubit in $V_\H$, we set $Q_i:=\Pi_i$. Otherwise, order the $k$ qubits on which $\Pi_i$ acts such that the
first one is in $V_\H$. $\Pi_i$ can be written as $\Pi_i=\ket{v_i}\bra{v_i} \otimes I_{n-k}$, where $\ket{v_i}$ is a
$k$-qubit state. We can decompose $\ket{v_i}=a_0 \ket{0} \otimes \ket{v^0_i}+a_1 \ket{1} \otimes \ket{v^1_i}$, where
the first part of the tensor product is the qubit in $V_\H$ and $\ket{v^1_i}$ and $\ket{v^2_i}$ are $(k-1)$-qubit
states on the remaining qubits. Define $Q_i=\ket{v^1_i}\bra{v^1_i}+\ket{v^2_i}\bra{v^2_i} \otimes I_{n-k+1}$. Call
$V_{{\LL}'}$ the set of qubits on which the projectors in ${\LL}'$ act on. Note that by construction $V_{{\LL}'}$ is
disjoint from $V_\H$, and that $V_{\H} \cup V_{{\LL}'}$ is the set of all qubits in ${\cal P}$; hence $\H$ and ${\LL}'$
are ``decoupled".

\begin{claim}
Assume there is a satisfying state $\ket{\Phi_{{\LL}'}}$ for ${\LL}'$ on $V_{{\LL}'}$. Then
$\ket{\Phi}=\ket{\Phi_{\H}}\otimes \ket{\Phi_{{\LL}'}}$ is a satisfying state for ${\cal P}$.
\end{claim}
\begin{proof}
By construction, $\ket{\Phi}$ satisfies all the projectors from $\H$ and all projectors in $\LL$ that do not have
qubits in $V_\H$. To see that it also satisfies any projector $\Pi_i$ from $\LL$ with a qubit in $V_\H$, observe that
$\ket{\Phi_{{\LL}'}}$ is orthogonal to both $\ket{v^1_i}$ and $\ket{v^2_i}$. Hence no matter how $\ket{\Phi_{{\LL}'}}$
is extended on the qubit of $V_\H$ in $\Pi_i$, the resulting state is orthogonal to $\ket{v_i}$.
\end{proof}
It remains to show that ${\cal L}'$ is satisfiable. This follows immediately from Cor.~\ref{cor:gen-sat-qlll}: we
observe that each projector in ${\LL}'$ can be viewed as a $k$-local projector of rank at most $4$; and by the
assumption each qubit in $V_{{\LL}'}$ appears in at most $2^{k}/(4 \cdot e \cdot k)$ projectors of ${\LL}'$.
\end{proof}

The Gluing Lemma~\ref{lem:glue} only depends on the underlying constraint hypergraph. We can hence give the
construction of the ``high degree" part of the instance purely in terms of hypergraphs, and will from now on associate
subsets of edges with the corresponding subsets of projectors. Motivated by the Gluing Lemma, our goal is to separate a
set of ``high degree" vertices $V_\H$ (above a certain cut-off degree $D$) with induced edges $\H$ such that each edge
outside $\H$ has at most one vertex in $V_\H$. We achieve this by starting with the high degree vertices and
iteratively adding all those edges that intersect in more than one vertex.

\begin{definition}[Construction of $V_{\H}$]\label{def:construct}
 Let $G=G([n],E)$ be a $k$-uniform hypergraph and $D>0$. Construct sets of vertices
$V_0,V_1,\ldots \subseteq [n]$ and edges $E_1, E_2\ldots \subseteq E$ iteratively in the following steps, starting with
all sets empty:
\begin{itemize}
\itemsep=-1pt \item[0)] Let $V_{0}=\{v \in V| \deg(v) > D\}$.
 \item[1)] For all $e \in E \setminus E_{0}$, if $e$ has $2$ or
more vertices in $V_{0}$, then add $e$ to $E_{1}$, and add to $V_1$ all vertices in $e$ not already in $V_{0}$.\\
$\vdots$

\item[i)] For all $e \in E \setminus (E_{0} \cup \ldots \cup E_{i-1})$, if $e$ has $2$ or more vertices in
$\bigcup_{j=0}^{i-1}V_{j}$, add $e$ to $E_{i}$, and add to $V_{i}$ all the vertices in $E_{i}$ which are not already in
$\bigcup_{j=0}^{i-1}V_{j}$. \item [] Stop at the first step $s$ such that $E_s=\emptyset$.
 \end{itemize}
 Let $V_{\H}:=\bigcup_{i=0}^{s} V_{i}$, ${\H}:=\bigcup_{i=1}^{s} E_{i}$ and $\LL:=E \setminus \H$.
\end{definition}
By construction all the $V_i$ are disjoint and similarly for the $E_i$. The process of adding edges stops at some step
$s$ ($E_s= \emptyset $), because $E \setminus (E_{0} \cup \ldots \cup E_{s-1})$ keeps shrinking until this happens.
Note that ${\H}$ consists precisely of all those edges in $E$ that have only vertices in $V_\H$ (i.e. $G(V_{\H},\H)$ is
the hypergraph induced by $G$ on $V_\H$).

To show that a random $\kqsat$ instance of density $\alpha$ is satisfiable with high probability, we only need to show
that the construction of $V_\H$, $\H$ and $\LL$ of Def.~\ref{def:construct} fulfills the conditions of the Gluing
Lemma~\ref{lem:glue} with high probability. We set $D=2^k/(4 \cdot e \cdot k)$ in Def.~\ref{def:construct}, so that
conditions 2. and 3. are fulfilled by construction. To finish the proof of Thm.~\ref{thm:threshold} it thus suffices to
show that the instance given by $\H$ on qubits in $V_\H$ is satisfiable. To show this we build on Laumann et al.'s
Thm.~\ref{thm:dimer}.

\begin{lemma}\label{lem:H_matches}
For a random $\kqsat$ instance with density $\alpha \leq  2^k/(12 \cdot e \cdot k^2)$, the reduced instance $\H$
obtained in the construction of Def.~\ref{def:construct} with $D=2^k/(4 \cdot e \cdot k)$ fulfills the matching
conditions of Thm.~\ref{thm:dimer} with high probability.
\end{lemma}
\begin{proof}
The proof of this key lemma proceeds in two parts. The first one (Lemma \ref{le:small_has_matching}) shows that {\em
any} hypergraph induced by a {\em small enough} subset of vertices in a hypergraph from $G_k(n,\alpha n)$ fulfills the
matching conditions. The second part (Lemma \ref{lem:H_is_small}) then shows that $V_\H$ is indeed small enough with
high probability.

\begin{lemma}[Small subgraphs have a matching]
\label{le:small_has_matching} Let $G$ be a random hypergraph distributed according to $G_{k}(n,\alpha n)$ and let
$\gamma=(e (e^2 \cdot \alpha)^{1/(k-2)})^{-1}$.  With high probability, for all $W\subset V$ with $|W| < \gamma n $,
the induced hypergraph on $W$ obeys the matching conditions of Thm.~\ref{thm:dimer}.
\end{lemma}

\begin{proof}
There is simple intuition why small sets obey the matching conditions - the density inside a small induced graph is
much smaller than the density of $G$: For simplicity set $\alpha=2^{k-1}$ and $\gamma=1/(2+2\delta)$ for some $\delta
>0$. Imagine fixing $W \subset V$ of size $\gamma n$ and then picking the graph $G$ according to $G_k(n,p)$ with
$p=\alpha n/ {n \choose k}\approx \frac{\alpha}{n^{k-1}}=(\frac{2}{n})^{k-1}$. The induced graph on $W$ is distributed
according to $G_k(\gamma n,p)$ and hence its density is $\alpha' = p\cdot {\gamma n \choose k}/\gamma n \approx p \cdot
(\gamma n)^{k-1}=(1+\delta)^{-(k-1)} \ll 1$. At such low densities the matching conditions are fulfilled with high
probability (see the remark below Thm.~\ref{thm:dimer}). We proceed to prove the somewhat stronger statement that the
matching conditions hold for {\em all} small subsets.

Let us first examine the matching conditions. We can construct a bipartite graph $B(G)$, where on the left we put the
edges of $G$ and on the right the vertices of $G$. We connect each edge on the left with those vertices on the right
that are contained in that edge. Then the matching conditions of Thm.~\ref{thm:dimer} are equivalent to saying that
there is a matching in $B(G)$ that covers all left vertices.

By Hall's theorem \cite{hall1935representatives,diestel1997graph}, such a matching exists iff for all $t$, every subset
of $t$ edges on the left is connected to at least $t$ vertices on the right. Hence, there is a ``bad" subset $W\subset
V$ with $|W| < \gamma n $ not obeying the matching conditions iff for some $t < \gamma n$ there is a subset of vertices
of size $t-1$ that contains $t$ edges. Let us compute the probability of such a bad event to happen.

First, fix a subset $S \subseteq V$  of size $t-1$ and let us compute the probability that it contains $t$ edges. The
probability that a random edge lands in $S$ is at most $((t-1)/n)^k$. Since in $G_k(n,m)$ all $m$ edges are picked
independently, we get
$$\Pr[S \,\, \textrm{contains}\,\, t \,\, \textrm{edges}] \leq {m \choose t} \left( \frac{t-1}{n} \right)^{kt}.$$
By the union bound over all subsets $S$ of size $t-1$ (there are ${n \choose t-1}$ of them) and all $t$ we get the
following bound
\begin{align*}
\Pr[ \exists \,\textrm{``bad" W}] & \leq \sum_{t=1}^{\gamma n} { n \choose t-1} {m \choose t}
\left(\frac{t-1}{n}\right)^{kt} \leq \sum_{t=1}^{\gamma n} { n \choose t} {\alpha n \choose t}
\left(\frac{t}{n}\right)^{kt}
\leq \left(\frac{ne}{t}\right)^{t}\left(\frac{\alpha n e}{t}\right)^{t} \left(\frac{t}{n}\right)^{kt}\\
&= \sum_{t=1}^{\gamma n} \left(e^2 \alpha \left(\frac{t}{n}\right)^{k-2}\right)^t=:\sum_{t=1}^{\gamma n} a_t.
\end{align*}
Note that the sum is clearly dominated by the first term ($t=1$). More precisely we have
\begin{align*}
\forall 1 \leq t < \gamma n-1 \quad \frac{a_{t+1}}{a_t}=e^2 \alpha
\left(\frac{t+1}{t}\right)^{(k-2)t}\left(\frac{t+1}{n}\right)^{k-2}\leq e^2 \alpha e^{k-2} \gamma^{k-2}=:r < 1,
\end{align*}
where for the last inequality we have used the bound on $\gamma$. Hence $\sum_{t=1}^{\gamma n} a_t \leq
\sum_{t=1}^{\gamma n} a_1 r^{t-1}=\frac{1}{1-r}a_1$, and we get $\Pr[\exists \,\textrm{``bad" W}] \leq \frac{1}{1-r}
\frac{e^2 \alpha}{n^{k-2}} \rightarrow 0$.
\end{proof}

\begin{lemma}[$V_{\H}$ is small]\label{lem:H_is_small}
Let $G$ be a hypergraph picked from $G_k(n,\alpha n)$ and let $V_{\H}$ be the set of vertices generated by the
procedure in Definition \ref{def:construct} with $D=2^k/(4 \cdot e \cdot k)$. Then for $k \geq 12$ and $\alpha k \leq
D/3$, with high probability $|V_{\H}|\leq (\epsilon_0+o(1)) n$ for some $\epsilon_0$ satisfying $\epsilon_0<\gamma$
where $\gamma$ is the constant from Lemma \ref{le:small_has_matching}.
\end{lemma}
\begin{remark}
As is standard in the model of random $\ksat$ and random $\kqsat$, if we look at the large $k$ limit we will always
first take the limit $n \rightarrow \infty$ for fixed $k$ and then $k \rightarrow \infty$. Hence we will always treat
$k$ (and $D$ and $\alpha$) as a constant in $O(\cdot)$ and $o(\cdot)$ terms.
\end{remark}
\begin{proof}
Throughout the proof we will set $\alpha$ to its maximum allowed value of $D/(3k)$. The statement of
Lemma~\ref{lem:H_is_small} for smaller $\alpha$ then follows by monotonicity.

For the proof of this lemma, we first replace $G_k(n, \alpha n)$ by a slightly different model of random hypergraphs
$G_k'(n, \alpha' n)$. In $G_k'(n, \alpha' n)$, we first generate a random sequence of vertices of length $k \alpha' n$
with each vertex picked i.i.d. at random. We then divide the sequence into blocks of length $k$ and, for each block
that contains $k$ different vertices, we create a hyperedge. (For blocks that contain the same vertex twice, we do
nothing.)

The expected number of blocks containing the same vertex twice is $O({k \choose 2} \alpha')=O(1)$. Therefore, we can
choose $\alpha'=\alpha+o(1)$ and, with high probability, we will get at least $\alpha n$ edges (and each of those edges
will be uniformly random). This means that it suffices to prove the lemma for $G_k'(n, \alpha' n)$.

For this model, we will show that $|V_i|$ satisfies the following bounds:

\begin{claim}\label{cl:bounds}
 There is an $\epsilon_0 < \frac{\gamma}{2}$ and $\epsilon_{i}:= 2^{-i}\epsilon_{0}$ such that for all
 $i: 0 \leq i \leq l$ with $l := \lceil \frac{3}{2} \log n \rceil$, with probability at least $1-\frac{2^i}{n^2}$,
\begin{equation}
\label{eq:req2}
  |V_i| \leq \epsilon_i n.
\end{equation}
\end{claim}

This implies that $V_l$ is empty with probability at least $1-O(\frac{1}{\sqrt{n}})$. In this case,
$|V_{\H}|=\sum_{i=0}^{l-1} |V_i|$. With probability at least $1-\frac{2^{l+1}}{n^2}=1-O(\frac{1}{\sqrt{n}})$,
(\ref{eq:req2}) is true for all $i$. Then,
\[ |V_H| = \sum_{i=0}^{l-1} |V_i|
\leq 2 \epsilon_0 n < \gamma n, \] which completes the proof of the lemma.

In what follows we will repeatedly use Azuma's inequality
\cite{azuma1967weighted,hoeffding1963probability,alon2004probabilistic}:

Let $Y_{0},\ldots,Y_{n}$ be a martingale, where $|Y_{i+1}-Y_{i}|\leq 1$ for all $0 \leq i < n$. For any $t > 0$,
\begin{equation}\label{eq:azuma}\pr(|X_{n}-X_{0}| \geq t) \leq \text{exp}(-\frac{t^{2}}{2n}). \end{equation}

We now prove Claim \ref{cl:bounds}, by induction on $i$.  We start with the base case $i=0$. Here, we will also bound
$R_0$, the number of edges incident to $V_0$, and show
\begin{equation}
\label{eq:req4}
 \Pr \left[ R_0 \geq \epsilon_0 D n \right] \leq \frac{1}{2n^2} .
\end{equation}

\paragraph{The $i=0$ case.} Recall that $V_{0}=\{v| \text{deg}(v) \geq D \}$. By linearity of expectation, $E[|V_{0}|]=n
\pr(\deg(v) \geq D )$. The degree of a vertex is a sum of independent 0-1 valued random variables with expectation
slightly less than $\alpha' k$. In the large $n$ limit, this becomes a Poisson distribution with mean $\leq \alpha' k
=D/3+o(1)$. Using the tail bound for Poisson distributions (see, e.g., \cite{alon2004probabilistic} Thm. A.1.15), we
obtain $\pr(\deg(v) \geq D )\leq (e^2/27)^{D/3}$. Note that for $k \geq 12$ we have
\begin{equation*} (\frac{e^2}{27})^{D/3}\leq \frac{5}{8}\epsilon_0, \mbox{~where~we~set~} \epsilon_0=\frac{\alpha'}{12
D^2 k}=\frac{D/(3k) + o(1)}{12 D^2 k} \leq \frac{1}{12 D k^2} < \frac{\gamma}{2}.\end{equation*} Then, $E[|V_0|]\leq
\frac{5}{8}\epsilon_0 n$.

To bound $E[R_0]$, observe that $R_0 \leq \sum_{v \in V_0} \text{deg}(v)$ and hence
\[ E[R_0] \leq n \Pr(\text{deg}(v)\geq D) \cdot E[\text{deg}(v)|\text{deg}(v) \geq D] \leq
\frac{5}{8}\epsilon_0 n \cdot E[\text{deg}(v)|\text{deg}(v) \geq D] \leq \frac{5}{6}\epsilon_0 n D,\]
 where for the last inequality we have used $E[\text{deg}(v)|\text{deg}(v) \geq D] \leq \frac{4}{3}D$, which follows from the following simple fact:
\begin{fact}
Let $X$ be a random variable distributed according to a Poisson distribution with mean $\lambda$. Then for $k>1$,
$E[X|X\geq k\lambda] \leq (k+1)\lambda$.
\end{fact}
\begin{proof}
\begin{align*}E[X| X \geq k\lambda ] &= \frac{\sum_{j=k\lambda}^{\infty} j\cdot \pr(X=j)}{\pr(X \geq k\lambda)} = \frac{1}{\pr(X \geq k\lambda)}
\sum_{j=k\lambda}^{\infty} j e^{-\lambda}\frac{\lambda^j}{j!}
 = \frac{1}{\pr(X \geq k\lambda)} \lambda \sum_{j=k\lambda}^{\infty}  e^{-\lambda}\frac{\lambda^{j-1}}{(j-1)!} \\
 &= \lambda \left(1+\frac{\pr(X=k\lambda-1)}{\pr(X\geq k\lambda)} \right)  \leq  \lambda \left(1+\frac{\pr(X=k\lambda-1)}{\pr(X = k\lambda)}
 \right)=\lambda \left(1+\frac{k\lambda}{\lambda}\right)=(1+k)\lambda.
 \end{align*}
 \end{proof}

To prove (\ref{eq:req2}) and (\ref{eq:req4}), we use Azuma's inequality Eq.~\eqref{eq:azuma}. Let $X_0, X_1, \ldots,
X_{ k\alpha' n}$ be the martingale defined in the following way. We pick the vertices of the sequence defining $G$ at
random one by one and let $X_i$ be the expectation of $|V_0|$ (resp. $R_0$) when the first $i$ vertices of the sequence
are already chosen and the rest is still uniformly random. Picking one vertex in any particular way changes the size of
$|V_0|$ by at most $1$ and of $R_0$ by at most $D$ (when the degree of a vertex crosses the threshold $D$ to be in
$V_0$). Therefore, for $V_0$, $|X_i-X_{i-1}|\leq 1$ ($|X_i-X_{i-1}|\leq D$ for the bound on $R_0$). For $V_0$, by
Azuma's inequality
\begin{equation*}
 \pr[| |V_0| - E[|V_0|] | \geq t] =
\pr[|X_{k \alpha' n}-X_0|\geq t] \leq e^{- \frac{t^2}{2 k \alpha' n}} .
\end{equation*} To make this probability less than $1/n^2$, we
chose $t= 2\sqrt{k \alpha'}  \sqrt{n \ln n}$. Then, with probability at least $1-\frac{1}{n^2}$, $|V_0| \leq
E[|V_0|]+O(\sqrt{n \log n})\leq \frac{5}{8}\epsilon_0 n+O(\sqrt{n \log n})\leq \epsilon_0n$, which gives
bound~\eqref{eq:req2}. Similarly, to show bound~\eqref{eq:req4} for $R_0$, we choose $t= 2D\sqrt{k \alpha'} \sqrt{n
(\ln n+1)}$. Then, we get that with probability at least $1-\frac{1}{2n^2}$, $R_0 \leq E[R_0]+O(\sqrt{n \log n})\leq
\frac{5}{6}\epsilon_0 nD+O(\sqrt{n \log n})\leq \epsilon_0nD$.

\paragraph{The $i>0$ case.} We will first {\em condition} on the event $\cal F$ that
bound~\eqref{eq:req4} holds and bounds~\eqref{eq:req2} hold for all previous $i$.
%By the union bound they do not hold
%with probability at most $O(\frac{i}{n^2}) = O(\frac{\log n}{n^2})$. Since $|V_i|$ cannot be larger than $n$, this case
%contributes at most $O(\frac{\log n}{n})$ to $E[|V_i|]$.
Moreover, we fix the following objects:
\begin{itemize}
\item The sets $V_0, \ldots, V_{i-1}$; \item The edges in $E_1, \ldots, E_{i-1}$; \item The degrees of all vertices
$v\in V_0 \union \ldots \union V_{i-1}$;
\end{itemize}

Conditioning on $V_0, \ldots, V_{i-1}$ and their degrees is equivalent to fixing the number of times that each $v\in
V_0 \union \ldots \union V_{i-1}$ appears in the sequence defining the graph $G$ according to $G_k'(n, \alpha'n)$.
Furthermore, conditioning on $E_1, \ldots, E_{i-1}$ means that we fix some blocks of the sequence to be equal to edges
in $E_1, \ldots, E_{i-1}$. We can then remove those blocks from the sequence and adjust the degrees of the vertices
that belong to those edges. Conditioning on $E_1, \ldots, E_{i-1}$ also means that we condition on the fact that there
is no other block containing two vertices from $V_0 \union \ldots \union V_{i-2}$.

We now consider a random sequence of vertices satisfying those constraints. Let $B$ be the total number of blocks
(after removing $E_0, \ldots, E_{i-1}$) and call $M_j$ the number of blocks that contain one element of $V_j$ for $0
\leq j \leq i-1$. (The $M_j$ are fixed since the $V_j$ are fixed.) The sequence of vertices on the $B$ blocks is
uniformly random among all sequences with a fixed number of occurrences of elements in $V_j$ (a total of $M_j$) and
such that no two of them occur in the same block. Note that an edge from $E_i$ must have at least one of its vertices
in $V_{i-1}$. We have $M_0+\ldots+M_{i-2}$ blocks containing one vertex from $V_0 \union \ldots \union V_{i-2}$ each.
For each of those blocks, the probability that one of the $M_{i-1}$ occurrences of $v\in V_{i-1}$ ends up in it is at
most
\begin{equation}
\label{eq:bound1} (k-1) \frac{M_{i-1}}{kB - M_0 - \ldots - M_{i-2}} .
%\onote{ not\ relevant\ anymore: (k-1) \frac{M_{i-1}}{n k \alpha - k(E_1+\ldots+E_{i-2}) - M_0 - \ldots - M_{i-2}} .}
\end{equation}
For any other block, the probability that two or more occurrences of $v\in V_{i-1}$ are in it is at most
\begin{equation}
\label{eq:bound2} {k \choose 2} \frac{M_{i-1}(M_{i-1}-1)}{(kB - M_0 - \ldots - M_{i-2})(kB - M_0 - \ldots - M_{i-2}-1)}
\leq {k \choose 2} \left( \frac{M_{i-1}}{kB - M_0 - \ldots - M_{i-2}} \right)^2 .
% \onote{not\ relevant\ anymore: \frac{k(k-1)}{2} \left( \frac{M_{i-1}}{n k \alpha - k(E_1+\ldots+E_{i-2})  - M_0 - \ldots - M_{i-2}} \right)^2 .}
\end{equation}
Observe that  $E_{j+1}+M_j \leq D V_j$ for $j \geq 1$ since each vertex in $V_j$ is incident to less than $D$ edges.
Moreover, $E_1+M_0 \leq R_0$. Note that this implies that $kB-(M_0 + M_1+\ldots + M_{i-2}) \geq k \alpha' n - k\left[
R_0 + D(V_1+\ldots + V_{i-2}) \right]$. Recall that we are conditioning on the event $\cal F$ that the bounds
in~\eqref{eq:req2} and~\eqref{eq:req4} hold, and hence we can further bound $kB-(M_0 + M_1+\ldots + M_{i-2}) \geq k
\alpha' n - k\left[ \epsilon_0 n D+ D(\epsilon_1 n+\ldots + \epsilon_{i-2}n)\right]\geq k \alpha' n-2kD\epsilon_0n$.
For our choice of $\epsilon_0 \leq \frac{\alpha'}{12kD^2}$ we hence obtain
\[ k B - (M_0 + \ldots + M_{i-2}) \geq \alpha' \frac{k}{2} n .\]
By combining (\ref{eq:bound1}) and (\ref{eq:bound2}), using the union bound for all relevant blocks, we get
\begin{align*}
 E[|E_i|] &\leq \left( (k-1) \frac{M_0+\ldots+M_{i-2}}{\alpha' \frac{k}{2} n} +\alpha' n {k \choose 2}
 \frac{ M_{i-1}}{(\alpha' \frac{k}{2} n)^2} \right) M_{i-1} \leq 2M_{i-1} \left( \frac{M_0+\ldots+M_{i-1}}{\alpha' n}
 \right)\\
& \leq 2D^2V_{i-1} \left( \frac{R_0/D+V_1+\ldots+V_{i-1}}{\alpha' n} \right).
\end{align*}
Since we are conditioning on the event $\cal F$ that~\eqref{eq:req2} and \eqref{eq:req4} hold, we can bound $R_0$ and
$V_j$ and obtain
\begin{equation*}
E[|E_i|]  \leq 2D^2V_{i-1} \left( \frac{\epsilon_0 n +\epsilon_1 n +\ldots+\epsilon_{i-1}n}{\alpha' n} \right)\leq 2D^2
V_{i-1}\frac{2 \epsilon_0}{\alpha'}\leq \frac{V_{i-1}}{3k},
\end{equation*}
where we have substituted $\epsilon_0=\frac{\alpha'}{12 D^2 k}$. Together with the observation that $|V_i|\leq k|E_i|$
we have hence shown in our setting that
 \begin{equation}\label{eq:expect-i}
 E[|V_i|]\leq V_{i-1}/3.
 \end{equation}
The large deviation bound (\ref{eq:req2}) again follows from Azuma's inequality~\eqref{eq:azuma}. We pick the sequence
of $kB$ vertices (after removing $E_0,\ldots,E_{i-1}$) vertex by vertex and let $X_i$ to be the expectation of $|V_i|$
after picking the $i$ first vertices of the sequence. Then, $X_0, X_1, \ldots, X_{kB}$ form a martingale and choosing
one vertex of the sequence affects $|V_i|$ by at most $k$. Therefore, $|X_i-X_{i-1}|\leq k$ when bounding $|V_i|$. We
now apply Azuma's inequality~\eqref{eq:azuma} with $t=2k\sqrt{k \alpha' n (\ln n+1) }$ and obtain, in our setting of
fixed sets $V_j$, fixed degrees of their elements and fixed sets $E_j$ for $0 \leq j \leq i-1$, and conditioning on the
event $\cal F$,
\[\pr(\left| |V_i|-E[|V_i|\right|] \geq O(\sqrt{n \log n}))\leq \frac{1}{2n^2}.
\]
Using Eq.~\eqref{eq:expect-i}, the induction hypothesis and the fact that we are conditioning on bound~\eqref{eq:req2}
to hold, we get that with probability at least $1-\frac{1}{2n^2}$, \[|V_i|\leq E[|V_i|]+O(\sqrt{n \log n})\leq
\frac{V_{i-1}}{3}+O(\sqrt{n \log n}) \leq \frac{\epsilon_{i-1}n}{3}+O(\sqrt{n \log n})\leq \epsilon_in.\] Since this
holds for all fixed sets $V_j$, fixed degrees of their elements and fixed sets $E_j$ for $0 \leq j \leq i-1$, it also
holds when we remove this conditioning (while still conditioning on the event $\cal F$). By the union bound, $\cal F$
does not hold with probability at most $\frac{2^i-\frac{1}{2}}{n^2}$. Hence, with probability at least
$1-\frac{2^i}{n^2}$, $V_i \leq \epsilon_i n$ and we have shown the bound in (\ref{eq:req2}).
\end{proof}

This terminates the proof of Lemma~\ref{lem:H_matches} for all  $k \geq 12$. For smaller values of $k$ our bound of
$\alpha \leq 2^k/(12 \cdot e \cdot k^2)$ is smaller than the bound obtained by Laumann et al. \cite{laumann09dimer},
and hence the Lemma also holds. Hence we have shown Thm.~\ref{thm:threshold}.
\end{proof}

\begin{remark}
Note that in the limit of large $k$ our results can be tightened to give a bound of $\alpha \leq (D-O(\sqrt{D \log
D}))/k=2^k/(4 \cdot e \cdot k^2)-O(\sqrt{2^{k/2}}\sqrt{k})$ for the satisfiable region. The analysis essentially
changes only for the bound on $E[|V_0|]$ in the beginning of the $i=0$ base case, where we have to use the tail bound
for the Poisson distribution for smaller deviations.
\end{remark}

\section*{Acknowledgments}
The authors would like to thank the Erwin Schr\"odinger International Institute in Vienna, where part of this work was
done, for its hospitality. We thank Noga Alon, Chris Laumann and Oded Regev for valuable discussions.

%\bibliographystyle{alpha}
%\bibliography{qlll}

\newcommand{\etalchar}[1]{$^{#1}$}

\end{document}